\documentclass[11pt]{article}
\usepackage{amsmath}

\def\qed{\relax\ifmmode\hskip2em \fbox{ }\else\unskip\nobreak\hskip1em 
$\fbox{}$\fi}

\newsavebox{\theorembox}
\newsavebox{\lemmabox}
\newsavebox{\corollarybox}
\newsavebox{\propositionbox}
\newsavebox{\examplebox}
\newsavebox{\propertybox}
\savebox{\theorembox}{\bf Theorem}
\savebox{\lemmabox}{\bf Lemma}
\savebox{\corollarybox}{\bf Corollary}
\savebox{\propositionbox}{\bf Proposition}
\savebox{\examplebox}{\bf Example}
\savebox{\propertybox}{\bf Property}
\newtheorem{theorem}{\usebox{\theorembox}}
\newtheorem{lemma}[theorem]{\usebox{\lemmabox}}
\newtheorem{corollary}[theorem]{\usebox{\corollarybox}}

\newtheorem{example}{\usebox{\examplebox}}
\newtheorem{definition}{{\sc Definition}\rm }[section]
\newtheorem{definitions}[definition]{{\sc Definitions\rm }}

\newenvironment{proof}{{\noindent\bf Proof~}}{\(\qed\)\vspace*{\proofskip} }
\newlength{\proofskip}
\let\oldbibliography\thebibliography
\renewcommand{\thebibliography}[1]{%
	\oldbibliography{#1}%
	\setlength{\itemsep}{0pt}%
}

\begin{document}

\begin{center}

{\LARGE 
Fast Approximation Schemes for Bin Packing
}


\footnotesize

\mbox{\large Srikrishnan Divakaran}\\
School of Engineering and Applied Sciences, Ahmedabad University, 
Ahmedabad, Gujarat, India 380 009,
\mbox{
srikrishnan.divakaran@ahduni.edu.in }\\[6pt]
\normalsize
\end{center}

\baselineskip 20pt plus .3pt minus .1pt


\noindent 
\begin{abstract}
We present new approximation schemes for   bin packing based on the following two             approaches: (1) partitioning     the given problem into         mostly identical sub-problems of constant size            and 
then construct a  solution by combining the solutions of  these constant size sub-problems obtained through PTAS or exact methods;    (2) solving   bin packing   using irregular sized  bins, a generalization of bin packing, that           facilitates the  design  of simple  and efficient recursive algorithms that  solve a  problem in terms of smaller sub-problems  such that
the    unused       space  in  bins used by an earlier solved sub-problem is    available to     subsequently solved sub-problems.  
\end{abstract}
\noindent\hrulefill
\section{Introduction}
The {\em Bin Packing} problem is a classical combinatorial optimization problem that was first studied in     the 1970's by Garey, Graham   and Ullman\cite{GGU72}         and Johnson\cite{J72}, and can be stated as  follows: 
\begin{quote}
 Given a collection $\cal{B}$ of unit capacity   bins and a    sequence $L= (a_1,a_2, ...,a_n)$   of  $n$  items   with their respective sizes $(s_1,s_2,...,s_n)$ such that $\forall{i}$ $s_i\in [0, 1]$,  determine a packing of the items in $L$ that uses a minimum number  of bins from $\cal{B}$. 
\end{quote}
Bin Packing has a  wide variety of applications\cite{JDUGG74} including cutting stock applications, packing           problems in supply chain management, resource     allocation problems in    distributed systems. Algorithms for bin packing         can be broadly    classified as {\em offline} and {\em online}. Offline         algorithms    are algorithms that pack items with complete knowledge of the  list $L$ of items prior to  packing, whereas online algorithms   need        to   pack items as they   arrive without any knowledge of future. The bin  packing problem even for the   offline version is known to be NP-Hard\cite{GJ79}   and hence has led researchers to the study of polynomial time approximation algorithms         (i.e. provides near optimal solutions). Most of  the initial  research in Bin Packing has been in the  design  of     simple deterministic algorithms and   their    combinatorial          analysis leading  to tighter upper bounds on the performance of these algorithms and tighter  lower bounds  on estimating the optimal offline and online solutions.        Subsequently, there has  been significant work     on probabilistic analysis of these deterministic algorithms       as  well as     on the design of randomized algorithms and approximation schemes for     Bin Packing. For a comprehensive survey of classical algorithms for Bin    Packing from the  perspective    of design and   analysis of approximation  algorithms, we       refer the readers to  Johnson's Phd Thesis\cite{J73}, Coffman et al.\cite{CCGMV13} and Hochbaum\cite{H97} .\newline \newline
 In 
 Bin Packing problem, we are required to pack the items
 in $L$ using minimum number of bins in $\cal{B}$.   In this paper, we  will refer to this classic version  of 
 Bin Packing as the {\em regular bin packing}  problem.
 The   offline  version  of    regular  bin    packing problem is known to be NP-Hard \cite{GJ79} and hence most research efforts have   focused on the the design of fast online and offline  approximation algorithms with good performance. The performance  of  an approximation algorithm is defined in terms of its worst case behavior  as follows: Let $A$ be an algorithm for  bin packing and let $A(L)$ denote the number of bins required       by $A$ to pack items in $L$, and OPT denote the optimal algorithm for    packing items in $L$. Let $\cal{L}$ denote the set of all possible list sequences whose items are of sizes  in $[0, 1]$. For every $k > 1$, $R_{A}(k) = sup_{L \in \cal{L}}^{} \{ A(L)/k  :  OPT(L)
  = k \}$. Then the asymptotic         worst    case  ratio is given by $R_{A}^{\infty} = lim_{k \rightarrow \infty} R_{A}(k)$.  This   ratio is    the  asymptotic approximation ratio  and measures   the quality of the algorithms packing in comparison to the optimal packing in the worst case scenario. The   second way of measuring the performance of an approximation algorithm      is  $sup_{L \in \cal{L}}^{} \{ A(L) / OPT(L)\}$   and this ratio is the absolute approximation ratio of the algorithm. In       the case of online algorithms this ratio is often referred to  as competitive ratio. \newline   
{\bf Online Algorithms}: $\text{NEXT-FIT(NF)}$,         $\text{FIRST-FIT(FF)}$  and $\text{BEST-FIT(BF)}$ are the three most natural online algorithms for regular  bin packing   that   has   been widely studied in the literature. These three algorithms   are a part of a larger class of algorithms    called {\em Any Fit (AF)} Algorithms   that at any time packs  an item into an empty bin only    if it does not fit into any already open bin. Johnson et al.\cite{J73, J74, JDUGG74}    showed that both  $FF$ and $BF$     have  an   asymptotic competitive ratio of $1.7$. Johnson  showed     that no $AF$ algorithm    can improve upon $FF$.  Yao's $\text{REVISED-FF(RFF)}$\cite{Y80}      was   the first  non-$AF$ online algorithm           with an  asymptotic competitive  ratio         of $5/3$ and    was based on $FF$ but    essentially classifies items into types based on        their sizes  and          uses separate bins for    different item types. Later    Lee and Lee\cite{LL85} generalized this idea        and  designed $Harmonic-Fit_{k} (HF_k)$ with asymptotic competitive ratio  $\approx  1.69103$. There   are many other           variants of $Harmonic$ and the best          among them is the algorithm of Seiden\cite{S02} and has an asymptotic   competitive ratio of $1.5889$. More recently,    Balogh et al.\cite{BBDSV15}     settled this online problem by presenting    an optimal online bin packing with absolute worst case competitive  ratio of $5/3$. \newline  {\bf Offline Algorithms}: The most natural offline   algorithms    first reorder the items and then employing other classical online algorithms   like $NF$, $FF$, $BF$ or other online      algorithms to pack the items. This  has   resulted    in three simple  but effective offline algorithms; they are denoted by $NFD$, $FFD$, and $BFD$, with the “D”           standing for “Decreasing". The sorting needs      $O(n log n)$ time and so the total running time of each of these    algorithms is $O(n log n)$.  Baker and   Coffman\cite{BC81} established    the asymptotic approximation ratio for $NFD$ to be $\approx 1.69103$, Johnson et al.\cite{JDUGG74}  established $FFD$ and $BFD$'s aymptotic      approximation ratio to       be $11/9$. Subsequently, Baker \cite{B85} and Yue\cite{Y91}. and Csirik\cite{C93} and Xu\cite{X00} presented simplified proofs. The first improvement over  $FFD$ was due to Yao's   $\text{Refined-First-Fit Decreasing (RFFD)}$\cite{Y80}  with an symptotic approximation ratio $=11/9- 10^7$.  This            was an $O(n^{10}logn)$ time        algorithm. Garey and Johnson\cite{GJ85} then   proposed      Modified First Fit (MFFD) , which   essentially packs the items with sizes in $(1/6, 1/3]$ after packing all items $>   1/3$, and then proved $R_{MFFD} = 71/60 = 1.183333$.  Friesen and Langsten\cite{FL91} also proposed two simple algorithms  $Best-Two-Fit (B2F)$ and $Combined Algorithm (CFB)$ that combines $B2F$          and $FFD$ with asymptotic approximation ratios $1.25$    and   $1.2$ respectively. \newline  
{\bf Asymptotic Approximation Schemes}: Fernandez de la Vega and Lueker \cite{FL81} presented     a PTAS that for any $\epsilon > 0$, designed an $C_{\epsilon} + Cnlog(1/\epsilon)$ time   algorithm $A$ with asymptotic worst case ratio $R_{A} \le 1 + \epsilon$, where $C_{\epsilon}$ and $C$ are constants that depend   on $\epsilon$. Johnson\cite{J82} observed  that if $\epsilon$ is             allowed to grow slowly  when compared to $OPT(L)$ then more efficient approximation            schemes can be constructed. This   was incorporated by Karmarkar and Karp\cite{KK82} to obtain an approximation  scheme where $A(L) \leq OPT(L) + O(OPT(L)^{1- \delta})$, for some positive constant $\delta$.   Using the idea of dual approximation algorithms, Hochbaum and Shmoys\cite{HS86, HS87}         present approximation schemes        for bin packing developed using polynomial approximation scheme for makespan.
\subsection{Our Results}
In this paper, we present fast polynomial time approximation  schemes for bin packing based on the following two approaches: (i)  {\em Near 
Identical Partitioning} and (ii) {\em Irregular    Bin Packing}.
Our approximation schemes     in some non-trivial special       cases yield asymptotically optimal solutions. \newline 
{\bf Near Identical Partitioning Approach}: In  this approach,  we  present an algorithm that (i) for  some real              number $\delta\in (0,\frac{1}{2})$ partitions      the input    sequence $L$ into $l$ identical sub-sequences $L_c$ (except for the last sub-sequence) of         length $c\in [\lceil \frac {1}{\delta} \rceil,\lceil\frac{2}{\delta }   \rceil]$ (i.e. sum of sizes of items in these subsequences is  $c$); (ii) determines     the optimal packing for $L_c$  and the last sub-sequence using  an      existing polynomial time approximation scheme for regular bin packing; and (iii) constructs   the packing for $L$ by concatenating the packing for the $l$ copies of $L_c$ and the          packing of the last sub-sequence.  \newline  
{\bf Irregular Bin Packing Approach}: In this approach,  we  view the regular bin packing problem as a special case       of {\em irregular bin packing}, a slight generalization, that can be defined as follows:
\begin{quote}
{\bf Irregular Bin Packing}: 
Given a sequence $L= (a_1,a_2,...,a_n)$ of $n$ items with their respective sizes $(s_1,s_2,...,s_n)$ such that $\forall{i}$ $s_i \in [0, 1]$, and    a collection $\cal{B}$ of $m$ bins with respective capacities   $c_1, c_2,..., c_m$ such that 
$\forall{i}$, $c_i \in [0, 1]$,   we   need   to     design an algorithm to pack
the items of $L$ among the bins in $\cal{B}$ that minimizes the opening/use   of {\em regular bins}, where    a   bin  is   {\em regular} if  its capacity is $1$ and {\em irregular} otherwise.
\end{quote}
Notice if all the bins in $\cal{B}$ are of unit capacity (i.e.   regular )  then the irregular bin packing reduces to the regular bin        packing  problem. In irregular    bin packing, if     an item is assigned to an irregular bin  (a bin with capacity $< 1$) then we   do    not charge the assigned item since that bin was already open. However, if        an item is assigned to a regular  bin  then we charge the assigned item $1$ unit  since it opens that bin.  
This     formulation   facilitates the  design  of simple and   efficient approximation schemes that recursively  solve  a problem in terms of smaller but similar sub-problems that  exploit the unused  space  in bins used by an earlier solved      sub-problem     by subsequently solved  sub-problems. \newline \newline
{\bf Paper Outline}: 
The rest of this 
paper is organized as follows: 
In Section $2$ we present a PTAS for bin packing through partitioning as described earlier, and in Section $3$ we present a PTAS for bin packing through a dynamic program for irregular bin packing.
\section{Bin Packing Through Near Identical Partitioning}
In this section, we     present    an algorithm that given a real valued 
parameter $\epsilon \in (0, \frac{1}{2})$, partitions the input sequence $L$ into  identical sub-sequences     (except for the last sub-sequence) of length $c \in [1, \lceil \frac {2}{\epsilon} \rceil]$ (i.e.   sum  of sizes of items in these subsequences is $c$) and then packs the items in these $c$-length     subsequences       onto  unit capacity bins    with  wastage (unused space) of at    most $\delta \in (\epsilon,\frac{1}{2})$
using   an existing polynomial time approximation scheme for regular bin packing.   Now, we introduce some    necessary terms and definitions and examples illustrating our key idea before presenting  our  approximation scheme and its analysis. 
\begin{definitions} 
	The sequence $L=(a_1,a_2, ..., a_n)$   with $k$ distinct   item sizes    \{$s_1, s_2, ..., s_k$\}       can       be       viewed as    a  $k$    dimensional vector $\hat{d}(L) = (n_1*s_1,n_2 *s_2,...,n_k*s_k)$, where for  $i \in [1..k]$, $n_i$ is the number of items of      type  $i$ (size $s_i$); we             refer to $\hat{d}(L)$  as  the distribution vector corresponding to $L$. For a   given real number $c > 1$, 
	let  $\hat{d}_{c}(L)$  denote a  $c$-length segment  of $\hat{d}(L)$ (i.e. a vector that is parallel to $\hat{d}(L)$ and contains its initial segment such that its component  sum equals   $c$).
\end{definitions} 

\begin{definitions}
	For a real number $\delta \in (0,\frac{1}{2})$, the  configuration of a unit capacity    bin containing items whose sizes are  $\{ s_1, s_2, ..., s_k \}$ and    has a     wastage of at most $\delta$ can be   specified by a $k$-dimensional  vector whose $i^{th}$ component, for $i\in [1..k]$, is the sum of   sizes of  items of type $i$ (size $s_i$) in that bin; and its length is    in the interval $[1-\delta, 1]$,  where  the length of      a           vector  is defined  to be the sum of  its components.  We refer to  such  a  vector  as a $(1 - \delta)$-vector (bin configuration)  consistent              with $L$; and we denote by $e_{\delta}(L)$ the set of all  ($1-\delta$)-vectors (bin configurations) consistent with $L$. 
\end{definitions} 
Note: For certain sequences $L$, the item sizes in $L$  may be such that for some $\delta \in (0, \frac{1}{2})$ there are no $1-\delta$ vectors consistent with $L$ (i.e. $e_{\delta}(L)$ is empty). 
\begin{definitions}
	For a given sequence $L$ and a real number $\delta\in (0,1/2]$, if $e_{\delta}(L)$ is non-empty then we define
	\begin{itemize} 
		\item [-]  a  $\delta$-cover for $\hat{d}(L)$ to be  a minimal collection of  $(1-\delta)$-vectors   from $e_{\delta}(L)$   such that for $i \in [1..k]$, the sum of the $i$th component of these collection  of vectors  is greater than or equal to the $i$th component of   $\hat{d}(L)$; 
		\item [-] $\text{min-cover}_{\delta}(\hat{d}(L))$ to be a $\delta$-cover for $\hat{d}(L)$ of the smallest size;
		\item [-] $\text{min-cover}(\hat{d}(L)) =\min_{\delta \in (0, 
			\frac{1}{2})}^{} \{ \text{min-cover}_{\delta}(\hat{d}(L)) \}$.  
	\end{itemize} 
\end{definitions}
{\bf Remark} : 
If    the  number of   distinct sizes in $L$ is not  bounded  by a constant $k$, then we   can   still  apply the above  idea by partitioning the  interval   $[0,1]$ into $k$   distinct     sizes   $0, 1/k, 2/k, ..., 1$ and    round  the   item  sizes  in  $L$   to  the nearest     multiple  of $1/k$ that is   greater than or  equal to the item size.  \newline \newline 
{\bf Key Idea}: For an integer $c^{*} \in [1, \lceil \frac{2}{\epsilon} \rceil]$, we partition  the distribution vector    $\hat{d}(L)$      into many copies of $\hat{d}_{c^{*}}(L)$, the $c^{*}$ length segment of 
$\hat{d}(L)$ (except for     the last segment), where $c^{*}$ is determined as follows:
For each $c \in [1, \lceil \frac{1}{\epsilon} \rceil]$, we determine 
$\delta_{c}$ to be a real number $\delta \in (\epsilon, 
\frac{1}{2})$ for which $\hat{d}_{c}(L)$ has the  smallest 
$\delta$-cover (i.e. $\text{min-cover}_{\delta_{c}}(\hat{d}_{c}(L)) =
\min_{\delta \in (\epsilon, \frac{1}{2})}^{} 
\text{min-cover}_{\delta}(\hat{d}_{c}(L))$). Then, we determine    
$c^{*}$ to be an integer in $[1, \lceil \frac{2} {\epsilon}\rceil]$ that
minimizes the packing ratio (ie.
$\frac{\text{min-cover}_{\delta_{c^{*}}}(\hat{d}_{c^{*}}(L))}{c^{*}} = 
\min_{c \ \in (1, \lceil \frac{2}{\epsilon} \rceil)}^{} 
\frac{\text{min-cover}_{\delta_{c}}(\hat{d}_{c}(L))}{c}$).
\begin{example}
	Let us consider a     sequence $L$ of $3000$ items consisting of $600$ items of size  $0.52$, $600$ items of size $0.29$, $600$ items of size $0.27$ and $1200$ items of  size $0.21$. Let $\epsilon = 0.1$ is the approximation ratio desired.  For this instance the distribution vector $\hat{d}(L)$ is a     $4$-dimensional vector $(0.21*1200,0.27*600,0.29*600,0.52*600)=(252,162,174,312)$    of  length $900$. Our  algorithm        attempts to  partition $\hat{d}(L)$ into         a $c$-segment vector for some  $c$                                        between $(1, \lceil\frac{2}{\epsilon}\rceil)$.  We can observe that we can partition $\hat{d}(L)$ into $60$ copies of the segment vector $(4.2,2.7,2.9,5.2)=(0.21*20, 0.27*10,0.29*10, 0.52*10)$ of length $15$. For $\delta \le 0.1$, the  minimum sized $\delta$ cover for this segment vector of length $15$  can be determined using any of the existing PTAS or exact algorithms for regular bin packing.
\end{example}
\begin{example}
	Let us consider a     sequence $L$ of $3000$ items consisting of $1000$ items of size  $0.60$, $1000$ items of size $0.65$,  and $1000$ items of  size $0.75$. Let $\epsilon = 0.1$ is the approximation ratio desired.  For this instance the distribution vector $\hat{d}(L)$ is a     $3$-dimensional vector $(0.60*1000,0.65*1000,0.75*1000)=(600, 650, 750)$    of  length $2000$. Our  algorithm        attempts to  partition $\hat{d}(L)$ into         a $c$-segment vector for some  $c$                                        between $(1,       \lceil \frac{2}{\epsilon}\rceil)$. We can observe      that  we can partition $\hat{d}(L)$ into $100$ copies of the segment vector $(6.0,6.50,7.50)=(0.60*10, 0.65*10,0.75*10)$  of length $20$. The           minimum sized $\delta$ cover for this segment vector of length $20$  can be determined using any of the existing PTAS or exact algorithms for regular bin packing. In this instance there 
	are no $\delta$-covers for $\delta < 0.4$. 
\end{example} 
{\bf ALGORITHM B($L$, $\epsilon$)}
\begin{tabbing}
	Input(s): \=  (1) \= $L$ \  = \=  $(a_1,a_2,...,a_n)$ \  be the sequence of $n$ items with  their     respective sizes  \\ 
	\>      \>         \> $(s_1,s_2, ... ,s_n)$ in the interval $[0, 1]$; \\
	\> (2) \> $\epsilon \in (0, \frac{1}{2})$ be a user specified parameter; \\
	Output(s): The assignment of the items in $L$ to the bins in $\cal{B}$; \\
	Begin \= \\ 
	\> (1) \= Let \= $\hat{d}(L) = (s_1*n_1, s_2*n_2, ..., s_k*n_k)$   be the distribution vector corresponding to $L$; \\
	\> (2) \= For \= ($c=1$; $c \le \lceil \frac{2}{\epsilon}\rceil$; $c= c+1$) \\
	\> (2a)    \>     \> Let $\hat{d}_{c}(L) = (s_1*n^{c}_1, s_2*n^{c}_2, ..., s_k * n^{c}_k)$  be the $c$-length segment of $\hat{d}(L)$; \\
	\> (2b)    \>     \> For \= ($\delta=\epsilon$; $\delta \le \frac{1}{2}$; 
	                         $\delta=\delta+1$) \\
	\>     \>     \>     \>
	                         $Cover_{\delta}(\hat{d}_{c}(L)) = 
							 \begin{cases}
	                         \text{min-cover}_{\delta}(\hat{d}_{c}(L))               & \text{if } e_{\delta}(L) \ne \Phi \\
	                         \Phi               & \text{otherwise} \\
	                         \end{cases}$ \\
	\>  (2c)   \>     \> Let $\delta_{c} \in (\epsilon, \frac{1}{2})$ be 
	a multiple of $\epsilon$ such that $|Cover_{\delta_{c}}(\hat{d}_{c}(L))|
	 = \min_{\delta \in (\epsilon, \frac{1}{2})}^{} 
	 |Cover_{\delta}(\hat{d}_{c}(L))|$  \\
	\>  (3) Let $c^{*}$ be an integer in $(1, \lceil \frac{2}{\epsilon} \rceil)$ such that  $\frac{|Cover_{\delta_{c^{*}}}(\hat{d}_{c^{*}}(L))|}{c^{*}} =  
	    \min_{c \ \in (1, \lceil \frac{2}{\epsilon} \rceil)}^{} 
	    \frac{|Cover_{\delta_{c}}(\hat{d}_{c}(L))|}{c}$; \\
	\>  (4) Let $T= \hat{d}_{c^{*}}(L)$  and $l = \frac{\hat{d}(L)}{c^{*}}$; \\
	\>  (5) \> Let $\text{Cover}(\hat{d}(L)) = \bigcup \limits_{i=1}^{l} \text{min-cover}(T) \cup \text{min-cover}(\hat{d}(L)-l*T)$; \\
	\> (6) \> return $\text{Cover}(\hat{d}(L))$ \\
	End 
\end{tabbing}
{\bf Determining $\text{min-cover}(\hat{d}_{c}(L))$}: For    determining 
the min-cover of $\hat{d}_{c}(L)$, we  need to determine     a  smallest sized collection of vectors from $e_{\delta}(L)$, $\delta \in [\epsilon, \frac{1}{2}]$, such   that for $i\in [1..k]$,    the sum of    the $i$th components of             these vectors is greater than or  equal to the $i$th component of    $\hat{d}_{c}(L)$. For this we  make use of the PTAS       result of Karmarkar and Karp\cite{KK82} to determine  a $\delta$-cover that is of size $(1 + \epsilon)|\text{min-cover}(T)|$ in polynomial time. We now introduce  some definitions that will help us present the analysis of 
Algorithm $B$. 
\begin{definitions}
	For notational convenience, let $T= (t_1, t_2,..., t_k)$ denote $\hat{d}_{c}(L)$, the $c$-length initial segment of $\hat{d}(L)$. 
	Let $\delta \in [\epsilon, \frac{1}{2}]$ be a real number and 
	$N=|e_{\delta}(L)|$          denote         the number of $1-\delta$ configurations consistent with $L$.  Let $C_1, C_2,..., C_{N}$ denote the complete enumeration of the $1-\delta$                  vectors (bins) consistent with $L$, where $c_{ij}$  denotes the  $i$th component of $C_j$. 
\end{definitions}  
Let $x_j$ denote  the number of bins packed according to  configuration $C_j$. Notice that the minimum $\delta$-cover for $T$ can be                solved using the PTAS for regular bin packing originally due to Fernandez de la Vega and Leuker\cite{FL81}, and later improved by Karmarkar and Karp\cite{KK82}. In this PTAS, the bin packing problem is  formulated as an integer program as follows:
\begin{equation} 
\text{minimize} \sum_{j=1}^{N} x_{j}
\end{equation} 
subject to 
\begin{equation} 
\sum_{j=1}^{N} c_{ij} x_{j} \geq t_i   \ \ \   i = 1, ..., k \\
\end{equation} 
\begin{equation} 
x_{j}\in N \ \ \  j=1, ..., N
\end{equation} 
\begin{definitions}
	Let Algorithm $B(L, \epsilon)$ partition $\hat{d}(L)$ into $l$ copies 
	of $T = \hat{d}_{c}(L)= (s_1*n^{c}_1, s_2*n^{c}_2,...,s_k*n^{c}_k)$ (discarding the last segment), 
	where $c$ is an integer in $[1, \lceil\frac{2} {\epsilon}\rceil]$ and
	$T$ is a $c$-length initial segment of $\hat{d}(L)$.
	Let $T^{t}=        (s_1 * \lfloor n^{c'}_1 \rfloor , s_2* \lfloor n^{c'}_2   \rfloor, ..., s_k * \lfloor n^{c'}_k \rfloor)$  be the segment vector  obtained by truncating for $i \in [1..k]$, the $i$th components of $T$ to the nearest integer multiple of $s_i$.
	Let $\text{Cover}(\hat{d}(L))$ be the $\delta$-cover determined by Algorithm $B$ for $\hat{d}(L)$. 
	\end{definitions}
\begin{theorem}
	$|Cover(\hat{d}(L))| \leq |\text{min-cover}(\hat{d}(L))| + k* l + 2c$.
\end{theorem}
The above theorem follows from Lemmas $3$, $4$ and $5$ presented 
below.
\begin{corollary}
	If (i) $T = (n_1*s_i, n_2*s_2, ..., n_k *s_k)$, where for $i \in [1..k]$ the $i$th component is an integer multiple of $s_i$ ; 
	OR (ii) $\sum_{i=1}^{k} s_i = o(c)$ OR $k =o(c)$, then Algorithm $B$ constructs an asymptotically optimal cover for $\hat{d}(L)$.
\end{corollary}
\begin{lemma}
	$|\text{min-cover}(T)| \leq |\text{min-cover}(T')| + 
	|\text{min-cover}(s_1, s_2, ..., s_k)| 
	\leq |\text{min-cover}(T')| + k$
\end{lemma}
\begin{proof}
	Notice that $T'$ is obtained by truncating each component $i \in [1..k]$, to the nearest multiple of $s_i$. 
	Therefore, 
	the maximum difference between the length of $T$ and $T'$ is $\sum_{i=1}^{k} s_i < k$. Therefore     the size of the optimal $\delta$-cover for $T$ cannot be more than the sum of the sizes of
	an optimal $\delta$-cover for $T'$ and an  optimal $\delta$-cover for $(s_1, s_2, ..., s_k)$. 
	For $i \in [1..k]$, the item sizes $s_i \in (0, 1)$. 
	Therefore the size of an optimal $\delta$-cover for $(s_1, s_2, ..., s_k)$ 	is at most $k$. 	Hence the result.
\end{proof}
\begin{lemma}
	$|\text{min-cover}(\hat{d}(L))| \geq l* |\text{min-cover}(T')|$
\end{lemma}
\begin{proof}
	Notice that $T$ was constructed by partitioning $\hat{d}(L)$ into identical segments of length $c \in [1, \lceil \frac{2}{\epsilon} \rceil]$ with the minimum packing ratio. Suppose $|OPT(\hat{d}(L))| 
	<  l* |\text{min-cover}(T')|$ then this would imply 
	that if we split $\hat{d}(L)$ into $l$ identical segment vectors 
	then at least one of these segment vectors would have a  packing ratio less than $T$.  A contradiction.
\end{proof}
\begin{lemma}
	$|\text{Cover}(\hat{d}(L))| \leq l*|\text{min-cover}(T)| + 2c$.
\end{lemma}
\begin{proof}
	The Algorithm $B$ splits $\hat{d}(L)$ into $l$ copies of segment vectors $T$ and a last segment vector $\hat{d}(L)-l*T)$ of length at most $c$. Therefore, 	the size of the minimum $\delta$-cover of the last segment vector $\hat{d}(L)-l*T)$  is at most $2c$. 
	Now by concatenating the min $\delta$-covers of $T$ $l$ times along with the min $\delta$-cover of the last segment we get the result.
\end{proof}
\section{Bin Packing Through Irregular Bin Packing}
In this section, we  present  a  recursive algorithm that can be converted into a  dynamic programming  solution to the irregular bin packing problem. Our              algorithm assumes (i) there   exists a   way  of  packing  the items   in $L$ using at   most $m$ unit bins; and (ii)  the  sizes of all    items in $L$     are integer multiples of a small positive   rational number $\delta$ less than $1$. \newline \newline 
Let $L=(a_1,a_2,...,a_n)$     be      a      sequence   of $n$ items with  their  respective    sizes $(s_1,s_2, ... ,s_n)$ in the interval $[0, 1]$ and $\cal{B}$ = \{$ B_1,B_2,...,B_m$ \}  be   a set of $m$ unit capacity bins. Let $\delta \in (0,1]$ be a rational number such that every item size in $L$ can be expressed as an integer multiple of $\delta$. Let  $D  = \lceil \frac{1}{\delta} \rceil$.  At  any   given instance, bins in $\cal{B}$ are classified based on its level   into one  of $D+1$  types : a bin            is  of  type  $i$ if its level is $i/D$, where      $i$        is an   integer in  $[0..D]$. Initially (at instance $0$), all     $m$    bins    are    of type $0$.  Our algorithm assigns   the items in $L$ one at a time onto bins  in $\cal{B}$    in the    order of their occurrence in $L$. The    state / configuration   of  our algorithm   at instance $i$ (i.e. after assigning  the $i$th item)      is specified in  terms of   a  $D+1$ tuple $C^{i} = ( n^{i}_0, n^{i}_1, n^{i}_2,..., n^{i}_{D})$,   where  $n^{i}_{j}$,  $j      \in [0,..D]$,  denotes     the   number  of  bins  of  type $j$. \newline \newline
{\bf Note}: We           classify       bins       based  on their level and not based on their          content. This  reduces the       number of  possible bin configurations and                          there     by  reducing the number of possible        sub-problems   that our dynamic program needs  to consider while  computing an optimal solution. \newline \newline 
Now, we  introduce     definitions that are necessary for presenting our dynamic 
program.
\begin{definitions}
For $i \in [1..n]$, we  define $type(a_i)=\lceil \frac{s_i}  {\delta} \rceil$ to 
be the type of item $a_i$, and          $a_i$   can be assigned to a bin of type $j$      only if  $s_i + j\delta \le 1$. For   a given      bin configuration $C = (n_0, n_1,n_2, ..., n_D)$ and     an item $a_i \in L$, we   define   $Allow(a_i, C) = \{ j : n_j \geq 1  \text { and } s_i + j\delta  \leq 1 \}$.           That is, $Allow(a_i, C)$  is the set of  bin types to which $a_i$  can        be assigned without  violating its capacity constraint.    For a given a bin   configuration $C$, we define $cost^{C}(a_i, j)$ to be $1$ if $a_i$ is assigned to an     empty 
bin of type $j$ in $Allow(a_i, C)$ and $0$ if it is assigned to a non-empty bin of type $j$ in $Allow(a_i, C)$. More formally, 
\begin{equation}
    cost^{C}(a_i, j) = \begin{cases}
               1               & j \in Allow(a_i, C) \text{ and is of type } 0\\
               0               & j \in Allow(a_i, C) \text{ and not of type } 0\\
               \infty          & otherwise 
                        \end{cases}
\end{equation}
\end{definitions}
{\bf Basic Description of Our Algorithm}: Our algorithm assigns the items    in $L$ to bins in $\cal{B}$ in order of their occurrence in $L$. The state  of our algorithm is defined in terms of the    configuration of bins in $\cal{B}$; the configuration of  a bin is defined in    terms of its  level and        not its composition ( the     type 
of items it contains). That is, while assigning an item, our algorithm does not  distinguish  between bins that are  filled   to    the   same  level but differ in their composition. Initially,  (at instance $0$), all   $m$ bins are of type $0$ (empty), so    the inital       state $C^{0}$ of our  algorithm is specified by the $D+1$ tuple $(m, 0, 0, ..., 0)$. Suppose        at    instance   $i-1$, our           algorithm is in state $C^{i-1} = (n^{i-1}_1, n^{i}_2, ...,  n^{i-1}_{D})$, where $n^{i-1}_{j}$,   $j \in [0,..D]$,  denotes   the number of bins of  type $j$. The next item  $a_i$ can  be assigned to  any bin whose type is     in $Allow(a_i, C^{i-1})$.     If our    algorithm  chooses     a bin  of  type   $j \in Allow(a_i,C^{i-1})$ then it          will end up in configuration $C^{i}_{j}$ and would cost $cost^{C^{i-1}}(a_i, j)$ plus the optimal  number of regular bins required for assigning the items in $L[i+1..n]$ starting        in configuration    $C^{i}_{j}$. So, our    algorithm assigns $a_i$ to a  bin type whose cost is  
$\min_{j \in Allow(C^{i-1})(a_i)}$\{  $Assign(C^{i}_{j}, i+1) + Cost^{C^{i-1}}(a_i,j)$\}
\begin{tabbing} 
{\bf PROCEDURE $Assign$($C$, $i$)} \\
Input(s): \= (1) \= $C = (n_0, n_1, n_2, ..., n_D)$ - the initial configuration of the $m$ bins in $\cal{B}$; \\
          \> (2) \> $i$ - the index of the next item in $L$ that needs to be assigned. \\
Output(s): \= The minimum number of regular bins required for assigning items in $L[i..n]$   \\
           \>     \> to bins in $\cal{B}$ in configuration $C$. \\
Begin \= \\
      \>     \= return \= $\min_{j \in Allow(a_i, C)}$ 
      \{ $Assign(C_{j}, i+1) + Cost^{C}(a_i,j)$  \}, where  \\
      \>     \> \> 
      $C_j=(n_0, n_1,...,n_{j-1},n_j-1,..., n_{o-1}, n_o+1, ..., n_{D})$
                                                and $o= j + type(a_i)$. \\
End                                                     
\end{tabbing}
\begin{lemma}
For regular bin packing problem, 
given any request sequence $L=(a_1,a_2,... ,a_n)$   whose item  sizes are 
in the interval $[0, 1]$ and are  integer multiples of a small   rational 
number $\delta \in (0,1)$ and      a  collection $\cal{B}$  of  $m$ unit capacity bins, the  procedure $Assign(C^{0}, 1)$ determines a   $(1+\delta)$ approximate  solution  in  approximately   $\frac{n}{\delta} m^{(\frac{1}{\delta}-1)}$ time, where
$C^{0} = (m, 0, 0, ..., 0)$ is a $D+1$-tuple donoting the initial 
bin configuration.  
\end{lemma}
\begin{proof}
Let $D =\frac{1}{\delta}$. The         run-time of $Assign(C^{0}, 1)$ is proportional to the number  of sub-problems, which in turn depends on the 
number         of           bin       configurations that our formulation 
permits. This  is the same as     the number of ways we can partition $m$ 
bins into $\frac{1}{\delta}$ categories   based on its level. This can be 
upper  bounded      by     ${m + D -1} \choose{D-1}$ =  $\frac{n}{\delta} m^{(\frac{1}{\delta}-1)}$.
\end{proof} 
\newline \newline 
{\bf Note}: If       we had       defined the bin configuration in terms 
of the  bin  composition (as usually done for PTAS for bin packing), then   
the number   of   bin types would be bounded by $R$=${M + K} \choose{M}$, 
where $K$ is the     number     of distinct  item sizes and the number of different bin  configurations is bounded by $P$=${n + R} \choose{R}$. So, 
defining      the state of the algorithm in terms of its level instead of 
its composition results in a significant reduction in the       number of 
states without impacting   its approximation guarantee.
\begin{theorem}
For a real $\epsilon \in (0, 1)$ and an integer $c > 1$, given a  request sequence  $L=(a_1,a_2, ... ,a_n)$ with item sizes in the interval $[0,1]$  
and  a  collection $\cal{B}$  of  $m$     unit capacity bins, the dynamic program  $A$($L$, $\cal{B}$) determines      a     $(1+\frac{\epsilon}{c} +\frac{1}{c})$ approximate  solution for regular  bin packing  problem in  approximately $\frac{nc}{\epsilon} m^{(\frac{c}{\epsilon}-1)}$
time.
\end{theorem}
\begin{proof}
Without    loss of generality we assume that all items in $L$ are  larger 
than $\epsilon$. Let $\delta = \epsilon/c$. First, we round the size   of 
each item $a_i$ in $L$ to  the smallest multiple of $\delta$ greater than 
or    equal to $s_i$. Let $L'$ be the modified instance of $L$. This will 
induce  a rounding error of at most $\delta =\frac{\epsilon}{c}$ for each 
item. Since each item is at  least $\epsilon$ in size. The rounding error 
is at most $\frac{1}{c}$. Now, if    we run the invoke      the algorithm 
$A$($L'$, $\cal{B}$) it will determine an optimal solution      for $L'$.
From Lemma $6$, we know that $A$  determines a   $(1+\delta)$ approximate  solution  for $L'$ in  approximately        $\frac{n}{\delta}m^{(\frac{1} {\delta}-1)}$ time. Since  the rounding error for converting $L$ to $L'$ 
is at most   $\frac{ \epsilon}{c}$. The approximation guarantee of its 
solution for $L$ would be  = $(1+\delta+\frac{1}{c})= (1+ \frac{\epsilon} {c}+\frac{1}{c})$ and the run-time would    be       $\frac{nc}{\epsilon} m^{(\frac{c} {\epsilon}-1)}$ time
\end{proof}

\end{document}